\documentclass[11pt]{article}

\usepackage[utf8]{inputenc}
\usepackage{amssymb,amsmath,amsthm,amscd,epsf,latexsym,verbatim,graphicx,amsfonts, mdframed}

\usepackage{tikz}
\usetikzlibrary{shapes}

\usepackage{float}

\usepackage{setspace}
\usepackage{enumitem}
\usepackage{tabularx}
\usepackage{supertabular}
\usepackage{array}
\usepackage{makecell}
\usepackage{booktabs}
\newcolumntype{x}[1]{>{\centering\arraybackslash}p{#1}}
\usepackage{algorithm}
\usepackage{algpseudocode}
\usepackage[justification=centering]{caption}
\usepackage{amsmath}
\usepackage{amsfonts}
\usepackage{amssymb}
\usepackage{graphics}
\usepackage{textgreek}
\usepackage[pdfencoding=auto,psdextra]{hyperref}
\usepackage{marginnote}
\usepackage{ulem} 

\definecolor{brass}{rgb}{0.71, 0.65, 0.26}

\newtheorem{theorem}{Theorem}
\newtheorem{lemma}[theorem]{Lemma}

\newcommand{\ddiag}{\mathrm{Diag}\,}

\newcommand{\Z}{\mathbb{Z}}
\newcommand{\R}{\mathbb{R}}

\newcommand{\bNP}{\mathbf {NP}}

\newcommand{\norm}[1]{\left\|#1\right\|}

\DeclareMathOperator{\Span}{span}

\newcommand{\rmv}[1]{}

\title{Bounded Distance Decoding for Random Lattices} 
\author{Shuhong Gao\\
School of Mathematical and Statistical Sciences\\
Clemson University\\
Clemson, SC 29634-0975\\
  Email: sgao@clemson.edu
}
\date{}

\begin{document}

\maketitle
\begin{abstract}

The current paper investigates the bounded distance decoding (BDD) problem for ensembles of lattices whose generator matrices have i.i.d. subGaussian entries. We first prove that, for these ensembles, BDD remains NP‑hard in the worst case. Then,  we introduce a polynomial‑time algorithm based on singular value decomposition (SVD) and establish—both theoretically and through extensive experiments—that, for a randomly selected lattice from the same ensemble, the algorithm solves BDD with high probability. To the best of our knowledge, this work provides the first example of a lattice problem that is NP‑hard in the worst case yet admits a polynomial‑time algorithm on average-case.

 {\bf Keywords: }  Lattices,  bounded distance decoding (BDD), the learning with error (LWE) problem,  NP-hard problems, average-case complexity, singular value decomposition (SVD).
  
\end{abstract}

\section{Introduction and Main Contributions}

 Let $m \geq n \geq 1$ be integers,  $\R$ the field of real numbers, and $\Z$ the ring of integers. The vectors in $\R^m$ and  $\Z^n$  are viewed as column vectors of length $m$ and $n$, respectively. 
A lattice $L$ of rank $n$ in $\R^{m}$ is generated  by a matrix $B \in \R^{m \times n}$ of  rank $n$, that is, 
\[ L = L(B) = \{Bx : x \in \Z^n\}, \]
a discrete additive subgroup of $\R^m$, called a lattice of rank $n$ in $\R^m$, and $B$ is called a basis of  $L$.
For any $x = (x_1, \ldots, x_m)^T \in \R^m$, let $\|x\| = \sqrt{\sum_{i=1}^m x_i^2}$ (the 2-norm of $x$), and for two vectors $b, v \in \R^m$, let $d(b, v) = \|b-v\|$ denote the distance between $b$ and $v$. 
The minimum distance and  the  determinant of $L$ are defined as 
\[ \lambda_1(L) = \min_{v \in L, \, v \ne 0} \| v \|,   \quad \quad  \det(L) = \sqrt{\det(B^T B)}.\]
By Minkowski's theorem, we know that  $\lambda_1(L) \leq \sqrt{n}\cdot  (\det(L))^{1/n}$.

\noindent {\bf Bounded Distance Decoding (BDD) Problem.} Given a matrix $B \in \R^{m\times n}$ of rank $n$, a target vector $b \in \R^m$ and a real number $r >0$, find $v \in L(B)$ so that $d(b, v) \leq r$ (if exists).
There are two possible choices for $r$:
\begin{enumerate}
\item[(a)] $r = \alpha \cdot \lambda_1(L(B))$ where $ 0 < \alpha < 1/2$ is a constant.  In this case, $\alpha$ is given and $\lambda_1(L(B))$ may be unknown. The BDD problem in this case is called  the $\alpha$-BDD problem in the literature   \cite{LLM06},  \cite{LM09}, \cite{BSW16}, \cite{BP20}, \cite{BPT22}.
\item[(b)] $r = c\sqrt{n}$ or $r = c \sqrt{n}  \cdot (\det(L))^{1/n}$ for some constant $c>0$.
\end{enumerate}
Case (a) arises naturally in practical applications, especially for  error-correction in coding theory where one is interested in  explicit constructions of 
 lattices (and linear codes over finite fields) with large minimum distances, and wants to find polynomial time decoding algorithms  for $\alpha < 1/2$.
 The minimum distance $\lambda_1(L)$ is typically known or pre-designed in these applications. 
 In complexity theory, one is often interested in the BDD problem for general classes of  lattices  where  $\lambda_1(L)$ is hard to compute exactly (or approximately). 
 Hence, for general lattices,  it is hard to verify whether  a given vector $v \in L$ is a valid solution for the $\alpha$-BDD problem.

The current paper focuses on Case (b). 
The primary objective  is to investigate  the worst-case and average-case complexities  of the BDD problem for various ensembles of lattices and for  $r = O(\sqrt{n})$ where $n$ is the rank of the lattices. 
The main  contribution  is encapsulated in the following three theorems, 
which demonstrate that there exist  ensembles  of lattices for which the BDD problem is NP-hard in the worst-case and yet there is a polynomial time algorithm  that can solve the BDD problem 
in polynomial time on the average-case in the sense that,  for a random lattice from the ensemble, the algorithm finds  a  required $v \in L$  with high probability in polynomial time.

The first theorem is about a simple method for reducing SAT-problem for Boolean variables to the BDD problem for lattices. 
\begin{theorem} \label{thm01}
For any formula  of the 3-SAT problem with $t$ clauses and  $k$ Boolean variables, there is a matrix $B \in \Z^{m\times n}$ of rank $n$ and a vector $b \in \Z^m$ with all  entries in $B$ and $b$ bounded by $2$,  where $m = k + 3t$ and $n = k + 2t$,   so that  the formula is satisfiable iff  $d(b, L(B)) \leq \sqrt{n}$. 
\end{theorem}

By the Cook-Levin Theorem,  the 3-SAT problem is NP-complete, hence the BBD problem is NP-hard for lattices defined by matrices with entries bounded by any constant $\gamma \geq 2$. 
Note that when $t<  k$, most formulae of the 3-SAT problem  with $t$ clauses over $k$ Boolean variables are satisfiable and it is easy to find a solution. Without loss of generality, we may assume $t \geq k$, hence 
$m = k +3t \geq \frac{4}{3}(k + 2t) =\frac{4}{3} n$. 

Later in the paper, we shall present a BDD algorithm based on singular value decomposition (SVD) for ensembles of random lattices that include the lattices from SAT problems as special cases. 
The lattice bases will be generated from Gaussian distributions and sub-Gaussian distributions. The next two theorems summarize what our BDD algorithm can achieve in success probability. 

\begin{theorem}\label{thm02}
Let $\sigma\geq 17$ be any constant and $\chi_{\sigma}$  the Gaussian distribution $N(0,\sigma^2)$ on $\R$. 
 Let  $m$ and $n$ be positive integers with $m\geq \frac{4}{3} n$. Given  a matrix $B \in \R^{m \times n}$ where each entry  is an i.i.d.\ sample from   $\chi_{\sigma}$,  and a vector $b \in \R^m$ with $d(b, L(B)) \leq \sqrt{n}$, our  BDD algorithm finds $v \in L(B)$ so that $d(b,v) \leq \sqrt{n}$ with probability 
at least $1-e^{-0.0045 m}$, where the success probability is over the randomness of $B$. 
\end{theorem}

One can prove a similar (weaker) result for more general ensembles of matrices that define random lattices. 
A random variable $X$ on $\R$  is said to be sub-Gaussian with
parameter  $\gamma>0$  if $E[X] =0$ and its moment generating function satisfies
\[ E[\exp(zX)]  \leq \exp\left(\frac{\gamma^2z^2}{2}\right), \quad  \text{for every }  z \in \R.\]
In terms of tail probabilities, if $X$ is sub-Gaussian with parameter $\gamma$  then, for every $z>0$, 
\[ P( X > z) \leq \exp\left(- \frac{z^2}{2 \gamma^2}\right), \quad \text{and}  \quad  P( X <- z) \leq \exp\left(- \frac{z^2}{2 \gamma^2}\right).\]
The variance $\sigma^2$ of a sub-Gaussian random variable with parameter $\gamma$  is at most $\gamma^2$. 
For examples, the  normal distribution $N(0, \gamma^2)$ is sub-Gaussian with parameter $\gamma$ and variance $\sigma^2=\gamma^2$,  and every random variable on the interval $[-\gamma, \gamma]$ with expected value $0$ is sub-Gaussian with parameter $\gamma$ and its variance varies. Particularly, the uniform random variable on the real interval $[-\gamma, \gamma]$ is sub-Gaussian with parameter $\gamma$ and  variance $\sigma^2 = \gamma^2/3$, and the Bernoulli uniform random variable on the set $\{-\gamma, \gamma\}$ is sub-Gaussian with parameter $\gamma$ and variance $\sigma^2 = \gamma^2$.

\begin{theorem}\label{thm02b}
Let  $\chi_{\sigma}$  be any sub-Gaussian distribution on $\R$ with variance $\sigma^2$.   There exist constants $c_1, c_2>0$, which depend only  on  the distribution $\chi_{\sigma}$, so that if 
\[ \frac{m}{n}  >  \left(1+ \frac{4 c_1}{\sigma}\right )^2, \]
then, for any  matrix $B \in \R^{m \times n}$  where each entry  is an i.i.d.\ sample from   $\chi_{\sigma}$ and  a vector $b \in \R^m$ with $d(b, L(B)) \leq \sqrt{n}$,  our BDD algorithm  finds $v \in L(B)$ satisfying $d(b,v) \leq \sqrt{n}$  with probability  at least  $1-2^{-(m-n+1)} -  e^{-c_2m}$, 
where the success probability is over the randomness of $B$.
\end{theorem}

While it remains a challenge to find the optimal values for the constants $c_1$ and $c_2$, we conducted computer experiments on the success probability of our BDD algorithm for subGaussian distributions.  Our  experimental  results shows that the practical performance of the BDD algorithm substantially exceeds the theoretical predictions by Theorem \ref{thm02b}, hence demonstrating that the BDD problem for the ensembles of lattice generated from any subGaussian distribution is solvable in polynomial time on the average-case, though it is NP-hard in the worst-case.

 The rest of the paper is organized as follows. Section \ref{sec2} is devoted to the proof of Theorem \ref{thm01},  Section \ref{sec3} gives the bounded distance decoding algorithm and  proof of Theorem \ref{thm02}, and Section \ref{sec4} presents data from computational experiments on the learning with error  (LWE) problem  over real numbers, which is a special case of the BDD problem.

\section{Reducing the SAT Problem to the BDD problem} \label{sec2}

This section is devoted to the proof of Theorem \ref{thm01}.
A Boolean variable $x$ takes value in $\{0, 1\}$ where $0$ represents  False and  $1$ for True, and the negation of $x$ is $\neg x = 1-x$. 
Let  $x = (x_1, x_2, ..., x_k)^T$  (a column vector) denote $k$ Boolean variables.  A  Boolean formula $\phi$ over $x$ is composed of literals (variables or their negation) and the logical connectives AND, OR, and NOT.   A Boolean formula is called  in conjunctive normal form (CNF) if it is a conjunction (AND) of one or more clauses, where a clause is a disjunction (OR) of literals. The SAT problem is,
given a formula in CNF in $k$ variables $x_1,\dots, x_k$, decide if  there exist $x_1,\dots, x_k\in\{0,1\}$ so  that the formula is true (i.e. takes value 1).
The next theorem can be found in most textbooks on computational complexity, see for example \cite{AB09}. 

\vspace{0.2cm}
\noindent {\bf Cook-Levin Theorem.}  {\it The SAT problem is NP-complete.}
\vspace{0.2cm}

\rmv{
This is one of the most important results in theoretical computer science, proved by Cook and Levin separately  around 1971.  Shortly after Cook's paper, Karp wrote a paper showing a list of 21 $\bNP$-complete problems. Now thousands of problems are known to be $\bNP$-complete.  These problems  are equivalent in the sense that if there is a polynomial time algorithm for one, then so is for every other one. 
}

We use a special case of the SAT problem. 
The 3-SAT problem is the SAT problem where each  formula consists 
of $\ell$-clauses with $1 \leq \ell \leq 3$. 
Note that any $\ell$-clause with $\ell >3$ can be converted a 3-SAT formula. Indeed, 
for any $\ell$-clause $(z_1\lor z_2 \lor \cdots \lor z_{\ell})$, introduce $\ell-2$ new variables $y_1, \ldots, y_{\ell-2}$. The $\ell$-clause can be transformed into a conjunction  of $3$-clauses:
\[ (z_1 \lor z_2 \lor y_1)\land (\neg y_1 \lor z_3 \lor y_2) \land (\neg y_2 \lor z_4 \lor y_3)\land \cdots \land (\neg y_{\ell-2} \lor z_{\ell-1} \lor z_\ell).
\]
One can check that the original clause is satisfied iff the same assignment to the original variables can be augmented by an assignment to the new variables that satisfy the conjunction  of clauses.
Hence, by  the Cook-Levin theorem, we see that  the 3-SAT problem is also $\bNP$-complete. 

Suppose $\phi$ is a formula  in CNF with $t$ clauses $C_1, \ldots, C_t$  in $k$ variables in $x = (x_1, \ldots, x_k)^T$.  Define an $t \times k$ matrix $A=(a_{ij})$ as follows where the $i$-th row represents the clause $C_i$ and the $j$-th column
 the variable $x_j$:
\[ a_{ij} = \left\{ \begin{array}{rl}   
1, & \text{ if $C_i$ contains $x_j$}, \\
0, & \text{ if $C_i$ does not contain $x_j$ nor $\neg x_j$},\\
-1, & \text{ if $C_i$ contains $\neg x_j$}.
\end{array}
\right.
\]
We call $A$ the incidence matrix of $\phi$. Also, define a vector $r = (r_1, \ldots, r_t)^T$ where  $r_i$ is the number of negations in $C_i$ for $1 \le i \le t$. 
For example, let 
$$\phi = (x_1 \lor x_2 \lor  x_3) \land (\neg x_1 \lor x_2 \lor x_4) \land (\neg x_2 \lor \neg x_3 \lor x_4) \land (x_1 \lor x_3 \lor \neg x_4) \land (\neg x_2 \lor \neg x_3 \lor \neg x_4).$$
Then 
\[ A = \begin{pmatrix}  1 & 1 & 1 & 0\\ -1 & 1 & 0 & 1\\ 0 & -1 & -1  & 1 \\ 1 & 0 & 1 & -1 \\ 0 & -1 & -1 & -1 \\ \end{pmatrix}, \quad  \quad 
r = \begin{pmatrix} r_1\\ r_2 \\ r_3 \\  r_4 \\ r_5 \\ \end{pmatrix} = \begin{pmatrix} 0\\ 1  \\ 2 \\ 1 \\ 3 \\ \end{pmatrix},\]
and 
\[ A x + r  = \begin{pmatrix} x_1 + x_2 + x_3 \\ (1-x_1) + x_2 + x_4  \\ (1-x_2) +(1-x_3) + x_4 \\ x_1 + x_3 + (1-x_4) \\ (1-x_2) + (1-x_3) + (1-x_4) \end{pmatrix}.\]
For any integer $t$, let $I_t$ denote the identity matrix of order $t$ and  $J_t$  the column vector of length $t$ with each entry equal to $1$. 
We claim that the formula $\phi$ is satisfiable iff there  exist $x \in \{0,1\}^k$ and  $y, z \in \{0,1\}^t$ so that 
\begin{equation}\label{eq_subset}
 A x + y - z  = 2J_t - r.
 \end{equation}
Indeed, suppose $\phi$ is satisfiable, that is, there exists $x \in \{0, 1\}^k$ so that $\phi$ is true, hence each clause $C_i$ is true, which means that  $v = Ax + r$ has each entry at least $1$ and at most $3$ 
(since each clause is a $\ell$-clause with $1 \leq \ell \leq 3$). For each $1 \leq i \leq t$,   let
\[
\begin{cases}
 y_i =1,  \quad z_i = 0, & \text{if } v_i = 1,\\
 y_i = 0,  \quad z_i = 0, & \text{if } v_i = 2,\\
 y_i = 0,  \quad z_i = 1, & \text{if } v_i = 3.
\end{cases}
\]
Then (\ref{eq_subset}) holds. Conversely, suppose  (\ref{eq_subset}) holds with some $x \in \{0,1\}^k, y, z \in \{0,1\}^t$. Since $|y_i - z_i| \leq 1$,  we see that  $Ax + r$ has each entry at least $1$, 
hence the formula $\phi$ is satisfied by $x$. 

Next,  let $n= k + 2t$. and  $m= k + 3t$. Define an $m \times n$ matrix $B$ and a vector $b \in \Z^m$: 
\begin{equation}\label{eq_SAT_CVP}
 B = \begin{pmatrix}    A&  I_t & 0 \\  2 I_k & 0 & 0 \\  0 & 2 I_t & 2I_t \\ 0 & 0 & 2 I_t  \end{pmatrix},
 \quad b = \begin{pmatrix} 2J_t - r \\ J_k \\ J_t \\ J_t \end{pmatrix}.
 \end{equation}
 Let $L = L(B)$, the lattice generated by the columns of $B$. Then $L$ has rank $n = k + 2t$ and each vector in $L$ is of the form
\[ v= B \begin{pmatrix} x \\ y-z \\ z\end{pmatrix} =  \begin{pmatrix} Ax +y -z  \\ 2x\\ 2y \\ 2z \end{pmatrix}\]
for some   $x \in \Z^k$ and some  $y, z \in \Z^t$. 
Note that, for any  $x \in \Z^k$,  we have $\|2x -J_k\|^2 \geq k$ and the equality holds iff  $x \in \{0,1\}^k$; similarly for $2y-J_t$ and $2z - J_t $. 
Hence we have $d^2(v,b) = \|v-b\|^2 \geq  k + 2t$ and the equality holds iff $x \in \{0,1\}^k, y,z \in \{0,1\}^t$,  and   (\ref{eq_subset}) holds. 
This proves that 
\[\min_{v \in L}  \norm{v -b} \geq \sqrt{k + 2t} = \sqrt{n},
\]
and the  equality holds iff the formulae $\phi$ is satisfiable.  In another word, the formula $\phi$ is satisfiable iff $d(b, L) \leq  \sqrt{n}$. 

This proves Theorem \ref{thm01}. By Cook-Levin Theorem,  this implies the BDD problem is NP-hard for lattices defined by matrices with entries bounded by any constant $\gamma \geq 2$.

\section{Bounded Distance  Decoding for lattices} \label{sec3}

In this section, we present our BDD algorithm and the proof of Theorems \ref{thm02} and \ref{thm02b}. 
In fact, we shall prove a  more general version where the distance bound can be any real number $r>0$, instead of $\sqrt{n}$. 

\subsection{SVD and  BDD  Algorithm}
The main tool we use  is SVD  (singular value decomposition) of matrices. 
Let $m \geq n+1$.  Every  matrix $M \in \R^{m\times (n+1)}$ has an SVD: 
\begin{equation} \label{eq_SVD}
 M = U \Sigma V^T
 \end{equation}
 where $U = (u_1, \ldots, u_m) \in \R^{m \times m}$ and $V = (v_1, \ldots, v_{n+1}) \in \R^{(n+1) \times (n+1)}$ are orthonormal  and  $\Sigma \in \R^{m \times (n+1)}$ is diagonal with entries
  $\sigma_1 \ge\sigma_2 \ge \cdots  \ge  \sigma_{n+1} \ge  0$ on its diagonal.   
  The matrices $U$ and $V$ are not unique in general, but the singular values  $\sigma_i$'s are unique, denoted $\sigma_i(M)$, which are characterized by the following property:
  for $1 \leq k \leq n+1$, 
  \begin{equation}\label{eq_MaxW2}
 \max_{W \subset \R^{n+1},\, \dim(W) =k\ }  \Big(\min_{x \in W, \ \|x\|_2 = 1}  \|Mx \|_2\Big) = \sigma_{k}(M), 
 \end{equation}
 where $W$ runs through linear subspaces of $\R^{n+1}$ and  the optimal value is obtained when $W = \Span(v_1, \ldots, v_k)$ and $x = v_k$; and 
\begin{equation}\label{eq_MinW2}
 \min_{W \subset \R^{n+1},\, \dim(W) =n+2-k\ } \Big( \max_{x \in W,\ \|x\|_2 = 1}  \|Mx \|_2 \Big) = \sigma_{k}(M), 
 \end{equation}
 where  the optimal value is obtained when $W = \Span(v_{k}, \ldots, v_{n+1})$ and $x= v_{k}$. This property will be useful for our  proof later.

Our DBB algorithm is presented in Figure (\ref{Fig01}). Its correctness and probability of success (for a random $B$) will be given in the next subsection. 
\begin{figure*}[th!]
   \begin{center}
   \renewcommand{\arraystretch}{1.25}
       \begin{tabular}{|p{1.5cm}p{12.5cm}p{.02cm}|}
       \hline
            & \textbf{Bounded Distance Decoding for Lattices} &  \\
 \hline
  Input: &   $B \in \R^{m \times n}$,  $b \in \R^m$, and  $r>0$. & \\ 
    Output: & $x \in \Z^n$ so that $\| Bx- b\| \leq r$, or "Failure". & \\
        \hline
Step 1. & Form the matrix $M :=(B, -b)$  and compute an SVD:   $M = U \Sigma V^T$ & \\
	    & where $U \in \R^{m \times m}$ and $V \in \R^{(n+1)\times (n+1)}$ are orthonormal  &\\
	    & and $\Sigma = \ddiag(\sigma_1, \ldots, \sigma_{n+1}) \in \R^{m \times (n+1)}$. &\\
Step 2. & Let $z=(z_1, \ldots, z_{n+1})^T$ be the $(n+1)$-st column of $V$. &\\
             & If $z_{n+1} = 0$ then return  "Failure". & \\
 Step 3.&   Compute $x = (x_1, \ldots, x_n)^T$ where   $x_i := \lceil \frac{z_i}{z_{n+1}}\rfloor \in \Z$ for $1 \leq i \leq n$. & \\
  &   If $\|Bx- b\| \leq r$  then return $x$ else return "Failure".  &\\
 \hline
      \end{tabular}
      \renewcommand{\arraystretch}{1}
      \caption{BDD Algorithm}
      \label{Fig01}
   \end{center}
\end{figure*}

\subsection{Theoretical Analysis}

The next lemma shows a connection between SVD and  short vectors in lattices. 
\begin{lemma}\label{lemma_SVD_Int}
Let $M\in \R^{m \times (n+1)}$  with  an SVD as in (\ref{eq_SVD}).  Let $k \leq n$ be an index  so that $\sigma_k > \sigma_{k+1}$. 
Then, for any $x \in \Z^n$ with   $\|Mx\|_2 <  \sigma_{k}/2$, there exists $z \in Z_k  = \Span(v_{k+1}, \ldots, v_{n+1})$ so that  $x  = \lfloor z  \rceil$. 
\end{lemma}
\begin{proof}
Decompose $x$ as   $x = z+ w$ where $z \in Z_k$ and $w \in \Span(v_1, \ldots, v_k)$.
We can write $z$ and $w$ as 
\[ z = \sum_{i=k+1}^{n+1} c_i v_i, \quad \quad  w = \sum_{i=1}^{k} c_i v_i\]
where $c_i \in \R$.  Since $V = (v_1, \ldots, v_{n+1})$ is orthonormal, we have  $\|w \|_2^2 = \sum_{i=1}^{k}  |c_i |^2$. Note that 
$Mx  =    \sum_{i=1}^{n+1} c_i\sigma_i  u_i$.
 Since $U = (u_1, \ldots,  u_m)$ is orthonormal, we have
\[ \|Mx \|_2^2 =   \sum_{i=1}^{n+1}  |c_i|^2   \sigma_i^2 \ge   \sum_{i=1}^{k}  |c_i|^2 \sigma_i^2
\ge   \sum_{i=1}^{k}  |c_i|^2 \sigma_k^2 =  \|w\|_2^2 \sigma_k^2.
\]
Since $\|Mx\|_2 <  \sigma_{k}/2$ by assumption,  we have  $\|w\|_2 \sigma_k <  \sigma_{k}/2$, thus $\|w\|_{\infty} \leq \|w\|_2 < 1/2$. 
 As $x \in \Z^n$, we have   $\lfloor z \rceil  = \lfloor x - w \rceil = x$ as claimed. 
\end{proof}

\begin{lemma}\label{Lemma_LWE}
Let $B \in \R^{m \times n}$,  $b \in \R^m$ and $x \in \Z^n$ so that 
$b = Bx + e$ for  $e \in \R^m$ with $\|e \|_2 \leq r$. 
  Suppose $\sigma_n(B) > 2 r$. Then
the BDD algorithm in Figure \ref{Fig01}  finds $x$ correctly. 
\end{lemma}
\begin{proof}
We first need to find good bounds for  the singular values of  $M = (B, -b)$, especially $\sigma_n(M)$ and $\sigma_{n+1}(M)$.
Let $\tilde{x} = \begin{pmatrix} x \\ 1 \end{pmatrix}  \in \R^{n+1}$.  Since  $M \tilde{x} = Bx - b   = -e$, we have $\|M \tilde{x}\|_2 = \|e\|_2 \leq r$, and 
\begin{equation}\label{eq_r}
 \sigma_{n+1}(M) \leq  \frac{\|M \tilde{x}\|_2}{\|\tilde{x}\|_2} \leq \frac{r }{\sqrt{1 + \|x \|_2^2}}  \leq  r. 
\end{equation}

Next we show that  $\sigma_n(M) \ge \sigma_n(B)$.  Let $U_0$ be the subspace of $\R^{n+1}$ of dimension $n$ with the last entry equal to $0$.  
For any $y = \begin{pmatrix} u \cr 0\end{pmatrix} \in U_0$ where $u \in \R^n$, we have
$My = Bu$. Hence 
\[ \sigma_n(B) = \min_{u \in \R^n, \, \|u\|_2 = 1}  \|Bu\|_2 =  \min_{y \in U_0, \, \|y\|_2 = 1}  \|My\|_2.\]
Since 
\[ \sigma_n(M)=  \max_{U \subset \R^{n+1},\,  \dim(U) =n\ }  \min_{y \in U, \, \|y\|_2 =1}  \|My\|_2,\]
we see that $\sigma_n(M) \geq \sigma_n(B)$. 

Since  $\sigma_n(B) > 2r$ by assumprtion,  we have  $\sigma_n(M)  >    2 r  \ge  2 \|M \tilde{x}\|_2$,  thus  (\ref{eq_r}) implies that    $\sigma_n(M)  >     \sigma_{n+1}(M)$. 
We can apply Lemma \ref{lemma_SVD_Int} to any SVD of  $M$:  $M = U \Sigma V^T$ with  $k =n$ and $Z_k = \Span(v_{n+1})$  where $v_{n+1}$ is the last column of $V$. This means that 
$\tilde{x} = \lfloor cv_{n+1}  \rceil$ for some $c  \in \R$. Let $v_{n+1} = (z_1, \ldots, z_{n+1})^T$. Since the last entry of $\tilde{x}$ is $1$,  we just need to run through
 real numbers $c$ so that $0.5 \leq c z_{n+1} \leq 1.5$ that gives difference integer vectors  
  $\tilde{x} = \lfloor c\cdot v_{n+1} \rceil$, and the number of such $c$ is bounded by $n \|v_{n+1}\|_{\infty}/|z_{n+1}|$.  
In our algorithm and computer experiments, we simply take $c$ to be $1/z_{n+1}$.  
\end{proof}

\begin{theorem}\label{thm03}
 Let  $m$ and $n$ be integers with $ m =n \beta$ where $ \beta >1$ is a constant. Suppose
  \[r >0, \quad  0< \epsilon< 1-\sqrt{\beta^{-1}},\quad  \sigma  >  \frac{2 r }{(1-\sqrt{\beta^{-1}} -\epsilon)\sqrt{m}}.\] 
Let $\chi_{\sigma}$  be  the Gaussian distribution  $N(0,\sigma^2)$ on $\R$ with variance  $\sigma^2$. 
For a random $B \in \R^{m \times n}$, where each entry  is an independent sample from   $\chi_{\sigma}$,  and for any vector $b \in \R^m$ with $d(b, L(B)) \leq r$, 
the BDD algorithm in Figure \ref{Fig01} runs in polynomial time and finds $x \in \Z^n$  so that $d(b,Bx) \leq r$ with probability 
at least  $1 - e^{-\epsilon^2 m/2}$, where the success probability is over the randomness of $B$. 
\end{theorem}

\begin{proof}
First, note that SVD can be computed in polynomial time \cite{DGH18,Golub2013}, hence our BDD algorithm runs in polynomial time.
Next, we need to examine the distribution of  the singular values of $B$. The well-known Bai-Yin Law (1993, \cite{BY93}) states that,
for any any probabilistic distribution $\chi$ on $\R$ with $E(\chi) = 0$, variance $E(|\chi|^2) = 1$ and $E(|\chi|^4) < \infty$, if $A$  is an $m \times n$ matrix with entries as i.i.d.\ from $\chi$, then  almost surely, 
\begin{equation}\label{eq_BY}
\lim_{m \to \infty} \frac{\sigma_n(A)}{\sqrt{m}}  =  1- \sqrt{\beta^{-1}}, \quad \quad   \lim_{m \to \infty} \frac{\sigma_1(A)}{\sqrt{m}}  = 1 + \sqrt{\beta^{-1}} .
\end{equation} 
There are more general concentration results, see for example  the book by Tao \cite{Tao12}.
For a Gaussian distribution, there is an effective version of the Bin-Yin Law. 
Suppose   $A$ is an $m \times n$ matrix with entries as i.i.d.\ from $N(0, \frac{1}{m})$. 
Davidson and Szarek (Theorem 2.13, \cite{DS01}) proved the following tail probability: 
\begin{equation}\label{eq_DS}
    P\left[ \sigma_n(A)  \le   1- \sqrt{\beta^{-1}} - t \right] <  e^{-\frac{mt^2}{2}}  \quad \text{ and } 
     P\left[ \sigma_1(A)  \ge  1+ \sqrt{\beta^{-1}} + t \right] <  e^{-\frac{mt^2}{2}}.
\end{equation}

Now back to our proof with $\chi_{\sigma}$ and $B$ as in Theorem \ref{thm03}. 
Let $A = \frac{1}{\sigma\sqrt{m}} B$. Then $A$ has entries  i.i.d.\  from  the Gaussian distribution $N(0, \frac{1}{m})$.  
We apply the inequality (\ref{eq_DS}) with $A$ and $t = \epsilon$, yielding 
the probability bound:
\[   P\left[ \sigma\sqrt{m} \cdot ( 1- \sqrt{\beta^{-1}} -\epsilon)  \le  \sigma_n(B) \le  \sigma_1(B) \le  \sigma \sqrt{m}\cdot  (1 + \sqrt{\beta^{-1}} +\epsilon)    \right] \ge  1- e^{-\frac{m\epsilon^2}{2}}. \]
By  the assumption on $\sigma$, we have $\sigma_n(B) > 2r$ with probability at least $ 1- e^{-\frac{m\epsilon^2}{2}}$. 
The rest of the proof is completed by  applying  Lemma \ref{Lemma_LWE}.
\end{proof}

{\bf Proof of Theorem \ref{thm02}}.  We show Theorem \ref{thm02} as a special case of Theorem \ref{thm03}. Indeed, let $\beta= 4/3$,  $r = \sqrt{n}$, and $\chi_{\sigma}$ a sub-Gaussian with variance $\sigma^2$ with $\sigma \geq  17$.  Noting  that  $1-\sqrt{3/4}  \approx 0.13397$, we let  $\epsilon = 0.03$. Then  $17>  2 \sqrt{\beta^{-1}}/ (1-\sqrt{\beta^{-1}} - \epsilon)  \approx  16.66$, hence 
\[ \sigma \geq 17   >  \frac{2r }{(1-\sqrt{\beta^{-1}} -\epsilon)\sqrt{m}}.\]
Then  Theorem \ref{thm03} says that, for a random $B \in \R^{m \times n}$, where each entry  is an independent sample from   
$\chi_{\sigma}$,  and for any vector $b \in \R^m$ with $d(b, L(B)) \leq r$, 
the BDD algorithm in Figure \ref{Fig01}  finds $x \in \Z^n$  so that $d(b,Bx) \leq r$ with probability 
at least  $1 -e^{-0.0045m}$.  This proves Theorem \ref{thm02}.

For a general subGaussian distribution, the following version of the Bin-Yin Law  is due to Rudelson and Vershynin (Theorem 1.1, \cite{RV09}). 
Let $\chi$ be any sub-Gaussian distribution on $\R$ with variance $1$.   
Suppose   $A$ is an $m \times n$ matrix with entries as i.i.d.\ from $\chi$. 
Then, for every $\epsilon > 0$ and $\beta = m/n >1$, we have
\begin{equation}\label{eq_RV}
    P\left[ \frac{\sigma_n(A)}{\sqrt{m}}   \le  \epsilon (1-\sqrt{\beta^{-1}} )  \right] \le  (c_1\epsilon)^{m-n+1} +  e^{-c_2m},
\end{equation}
where $c_1, c_2>0$ are constants depending only  on  the distribution $\chi$. 
This bound and the above proof of Theorem \ref{thm03} immediately lead to the following theorem.

\begin{theorem}\label{thm04}
 Let  $m$ and $n$ be integers with $ m =n \beta$ where  $\beta >1$ is a constant. Suppose
  \[r >0, \quad  \epsilon>0,\quad  \sigma >  \frac{2 r }{\epsilon  (1-\sqrt{\beta^{-1}})\sqrt{m}}.\] 
Let $\chi_{\sigma}$  be  a subGaussian distribution on $\R$ with variance  $\sigma^2$. 
For a random $B \in \R^{m \times n}$, where each entry  is an independent sample from   $\chi_{\sigma}$,  and for any vector $b \in \R^m$ with $d(b, L(B)) \leq r$, 
the BDD algorithm in Figure \ref{Fig01} runs in polynomial time and finds $x \in \Z^n$  so that $d(b,Bx) \leq r$ with probability 
at least  $1-(c_1\epsilon)^{m-n+1} -  e^{-c_2m}$, where the success probability is over the randomness of $B$ and  
$c_1, c_2>0$ are constants depending only  on  the distribution $\chi_{\sigma}$. 
\end{theorem}

Theorem \ref{thm02b} in the introduction is a special case of Theorem \ref{thm04} when $\epsilon = 1/(2c_1)$ and $r = \sqrt{n}$. 
Determining optimal values for the constants $c_1$ and $c_2$ remains an open challenge. However, we note that Hsu, Kakade, and Zhang (Lemma A.1 in \cite{HKZ12}) derived an explicit bound on the largest singular value  $\sigma_1(B)$ with explicit  constants.  Their method could be used to derive explicit  bounds for the smallest singular value $\sigma_n(B)$. 
We direct readers to their paper for further technical details. In the subsequent section, we present computational experiments to evaluate the success probabilities, demonstrating that the practical performance of our BDD algorithm  significantly surpasses the theoretical predictions provided by these bounds.

\section{Computer Experiments on the LWE Problem over Real Numbers} \label{sec4}

We evaluate the performance of the BDD algorithm in solving the learning with errors (LWE) problem over real numbers, studied by Bootle et al.\  \cite{BDE18}.  
For any $\gamma>0$, let $U_{\gamma}$ denote the uniform distribution on the real interval $[-\gamma, \gamma]$, which is sub-Gaussian with parameter $\gamma$ and variance $\frac{\gamma^2}{3}$.   
Let $\beta >1$,  $\theta >0$  and $\gamma_0>0$ be constants, and
\begin{equation}\label{eq_mg}
 m= \lceil \beta n\rceil, \quad \gamma_1 =\theta  \gamma_0.
 \end{equation}
 The LWE problem over real numbers is to find $x \in \Z^n$, given $B \in \R^{m \times n}$  and  $b \in \R^m$, such  that
\begin{equation}\label{eq_lwe2}
 b = Bx + e, 
 \end{equation}
where each entry of $B$  is an independent sample from $U_{\gamma_1}$ and 
 each entry of $e$ is an an independent sample from $U_{\gamma_0}$.    Bootle et al.\  \cite{BDE18} proved that this problem can be solved by the least square solution when $m$ is sufficiently large. They studied the LWE problem over reals for more general subGaussian distributions and gave an explicit bound on $m$ in their Theorem 4.5, which says that  the ratio $\beta= m/n$ should be  at least $C_1=2^8\log 9$.

  We would like to demonstrate that our SVD method works well for the case $\gamma_0 < \gamma_1$. It works not just  when the ratio $\beta=m/n$ is large, but also for small $\beta$, say $\beta  \leq 2$.  
 We used the Multiprecision Computing Toolbox for MATLAB \cite{mct23} for computing SVD.
 In our computer experiments, we generate sample data $(B,b)$ satisfying (\ref{eq_lwe2}) as follows. For any given $(n, \beta, \gamma_0, \theta)$, 
 we calculate $m$ and $\gamma_1$ as defined in (\ref{eq_mg}), select $x$ uniformly from $\{0, 1\}^n$, and generate each entry of $B$ and $e$ as an independent sample from $U_{\gamma_1}$ and $U_{\gamma_0}$, respectively. We then rescale $(B, b)$ so that $\gamma_0 = 1$. 
 For each  $(n, \beta,  \theta)$  in Table \ref{data1}, we generated 1000 random samples of $(B, b)$, recording the number of successful runs where the BDD algorithm correctly recovered $x$. The success probability is calculated as the number of successful runs divided by 1000. Results in the "Prob" column show that, for $\theta = 2$ (or higher) and $\beta \geq 1.2$, the BDD algorithm achieves at least a 97\% success rate, and for $\beta = 1.5$ (or higher) with $\theta \geq 1.1$, the success rate is at least 95\%. These findings suggest that the LWE problem over real numbers can be solved efficiently on average.

 \begin{table}[h]
\centering
\begin{tabular}{|c|c|c|c|c|c||c|c|c|c|c|c|}
\hline
\( n \) & \( \beta \) & \( \theta \) & \( \text{Prob} \) & \( \text{Prob*} \) & 
\( n \) & \( \beta \)  & \( \theta \) & \( \text{Prob} \) & \( \text{Prob*} \)  \\ \hline
100 & 1.0  & 2 & 0.007 & 0.000 & 100 & 1.5  & 0.7 & 0.088 & 0.026  \\ \hline
100 & 1.1  & 2 & 0.723 & 0.242 & 100 & 1.5  & 0.9 & 0.647 & 0.395 \\ \hline
100 & 1.2  & 2 & 0.979 & 0.740 & 100 & 1.5  & 1.1 & 0.957 & 0.678  \\ \hline
100 & 1.3  & 2 & 1.000 & 0.966 & 100 & 1.5  & 1.3 & 0.997 & 0.826 \\ \hline
100 & 1.4  & 2 & 1.000 & 0.996 & 100 & 1.5  & 1.5 &1.000  & 0.871 \\ \hline
100 & 1.5  & 2 & 1.000 & 0.999 & 100 & 1.5  & 1.7 & 1.000 & 0.991 \\ \hline
100 & 1.6  & 2 & 1.000 & 1.000 & 100 & 1.5  & 1.9 & 1.000 & 0.999 \\ \hline
100 & 1.7  & 2 & 1.000  & 1.000 & 100 & 1.5 & 2.1 & 1.000  & 1.000 \\ \hline
\end{tabular}
\caption{The success probability for the BDD algorithm, where $\gamma_0 =1$.\\
The column of Prob is for the LWE problem  over reals, and \\
the column of Prob* is  for the LWE problem over integers.}
\label{data1}
\end{table}

 We also conducted experiments on the integer version of the LWE problem. The "Prob*" column reflects the success probability for this version. For each $B$ and $e$ sampled from $U_{\gamma_1}$ and $U_{\gamma_0}$, we rounded each entry to the nearest integer and computed $b$ as in (\ref{eq_lwe2}). With $\gamma_0 = 1$, each entry of $e$ is in $\{-1, 0, 1\}$ with probabilities $1/4, 1/2$, and $1/4$, respectively. When $\gamma_1 = \theta \leq 2.1$, the entries of $B$ are integers bounded by $\theta$, distributed as follows:
 \begin{enumerate}
\item[(a)]  For $0.7 \leq \theta \leq 1.5$: entries $-2$ and $2$ have probability 0,  while $-1$ and $1$ each have probability $\frac{\theta - 0.5}{2 \theta}$, and $0$ has probability $\frac{1}{2 \theta}$.
\item[(b)]   For $1.5 < \theta \leq 2.1$: entries $-2$ and $2$ each have probability $\frac{\theta - 1.5}{2 \theta}$, and $-1$, $0$, and $1$ each have probability $\frac{1}{2 \theta}$.
\end{enumerate}
When $\beta \geq 1.5$ and $\theta \geq 1.1$, results in the "Prob*" column show that the BDD algorithm achieves a success rate of at least 67\%.

The above experiment data indicate the BDD problem for lattices generated from sub-Gaussian distributions can be solved in polynomial time on the average-case, and the BDD decoding algorithm perform much better than the theoretical bounds in Theorem \ref{thm04} might indicate.

\section{Conclusions}

In conclusion, we presented an efficient algorithm that solves the BDD problem for ensembles of lattices generated by matrices from Gaussian and sub-Gaussian distributions. In the worst case, the BDD problem for these  lattices remain hard to solve unless NP = P; however, on the average case, our BDD algorithm solves them in polynomial time. Our experimental results indicate that the BDD algorithm performs much better than the theoretical bounds may indicate. It remains an open problem to find the optimal values for the constants $c_1$ and $c_2$ for sub-Gaussian distributions. 

It is worth noting that our algorithm does not perform well on lattices generated by matrices associated with the 3-SAT problem, nor does it appear effective for the LWE problem over finite fields. Given that several post-quantum cryptographic schemes, including recent standards published by NIST, are lattice-based, further work is needed  to evaluate the potential impact on these PQC schemes.

\bigskip
\noindent \textbf{Acknowledgement:} This work is based upon the work supported by the National Center for Transportation Cybersecurity and Resiliency (TraCR) (a U.S. Department of Transportation National University Transportation Center) headquartered at Clemson University, Clemson, South Carolina, USA. Any opinions, findings, conclusions and recommendations expressed in this material are those of the author(s) and do not necessarily reflect the views of TraCR, and the U.S. Government assumes no liability for the contents or use thereof.

\bibliographystyle{plain}

\bibliography{Lattice_BDD_Ref.bib}

\end{document}